\documentclass{article}

\usepackage{arxiv}

\usepackage{tikz}
\usetikzlibrary{automata}
\usepackage{graphicx}
\usepackage{booktabs}
\usepackage{amsthm,amsmath,latexsym,bbm,enumitem,xcolor}

\newtheorem{theorem}{Theorem}
\newtheorem{lemma}{Lemma}
\newtheorem{definition}{Definition}

\newtheorem{proposition}{Proposition}

\title{System Stability Under Adversarial Injection of Dependent Tasks} 

\author{
  Vicent  Cholvi \\
  Departament de Llenguatges i Sistemes Inform\`{a}tics\\
  Universitat Jaume I\\
   Castell\'{o}, Spain \\
   \and
 Juan Echag\"{u}e \\
  Departament de Llenguatges i Sistemes Inform\`{a}tics\\
  Universitat Jaume I\\
   Castell\'{o}, Spain \\
   \and
   Antonio Fern\'andez Anta\\
  IMDEA Networks Institute\\
  Madrid, Spain\\
   \and
   Christopher Thraves Caro\\
   Depto. Ing. Mat. \\
   Facultad  de Ciencias F\'{\i}sicas y Matem\'aticas\\
   Universidad de Concepci\'on, Chile
}

\newcommand{\cA}{{\mathcal A}}

\newcommand{\blue}[1]{{\color{blue}{#1}}}
\newcommand{\red}[1]{{\color{red}{#1}}}
\newcommand{\cyan}[1]{{\color{cyan}{#1}}}
\renewcommand{\blue}[1]{{{#1}}}
\renewcommand{\red}[1]{{{#1}}}
\renewcommand{\cyan}[1]{{{#1}}}

\begin{document}

\maketitle

\begin{abstract}
In this work, we consider a computational model of a distributed system formed by a set of
servers in which jobs, that are continuously arriving, have to be executed. Every job is formed by a set of
dependent tasks (i.~e., each task may have to wait for others to be completed before it can be started), each of which has to be executed in one of the servers. The arrival of jobs and their properties is assumed to be controlled by a bounded adversary, whose only restriction is that it cannot overload any server. This model is a non-trivial generalization of the Adversarial Queuing Theory model of Borodin et al., and, like that model, focuses on the stability of the system: whether the number of jobs pending to be completed is bounded at all times. We show multiple results of stability and instability for this adversarial model under different combinations of the scheduling policy used at the servers, the arrival rate, and the dependence between tasks in the jobs.
\end{abstract}


\section{Introduction}

In this work, we consider a model of jobs formed by dependent tasks 
that have to be executed in a set of servers. 
\cyan{The dependencies among the tasks of a job restrict} the order and time of their execution. 
For instance, a task $q$ may need some information from 
another task $p$, so that the latter must complete before $q$ can be
executed.
%
This model embodies, for instance, the dynamics of Network Function 
Virtualization (NFV) systems \cite{herrera2016resource,yi2018comprehensive} or Osmotic Computing (OC) \cite{villari2016osmotic}. In a NFV system,
network services (which are job types) are specified as 
service chains, obtained by the concatenation of network functions. 
These network functions are dependent computational tasks to be executed in the NFV
Infrastructure (e.g., servers distributed over the network). In an OC system, an application is divided into microservices that are distributed and deployed on an edge/cloud server infrastructure. The user requests (jobs) involve processing
\cyan{(tasks) in several of these microservices,} as defined by an orchestrator that takes into account the dependencies between the
microservices. In that line, it also  encompasses a number of features of Orchestration Languages (see, for instance,~\cite{Kitchin2009}), which propose a way to relate concurrent tasks to each other in a controlled fashion: the invocation of tasks to achieve a goal, the synchronization between tasks, managing  priorities, etc.

In our model, we consider a dynamic system in which job requests (or jobs for short) are
continuously arriving. Each job contains the whole specification of 
its dependent tasks: the collection of tasks to be executed, the server that 
must execute each task, the time the execution incurs, 
\cyan{the dependencies among tasks,}
etc.
Instead of assuming stochastic job 
\cyan{arrivals into the system, 
in our model we assume the 
existence of an adversary 
that has full control of the job} 
requests arrivals, and the specification of their tasks. The 
only restriction on the adversary is that no server can be overloaded in 
the long run (some burstiness in the load is allowed).

In this adversarial framework, the objective is to achieve stability in 
the system. This means that the system is able to
\cyan{cope with the adversarial arrivals, maintaining the number of pending job requests in the 
system bounded at all times. (This usually also implies that all the job 
requests are eventually completed.)}
Observe that the framework assumes that the resource allocation 
is done by the adversary (since it chooses where tasks have to be 
executed and in which order). Hence, the only tool we have to achieve 
stability is the scheduling of tasks in the different servers.

The study of the quality of service that can be provided under worst-case 
assumptions in a given system (NFV or OC, for instance) is important in order to 
be able to honor Service Level Agreements (SLA). The positive 
results we obtain in this paper show that it is possible to guarantee a 
certain level of service even under pessimistic assumptions. These 
results can also be used to separate resource allocation and scheduling 
as long as the resource allocation guarantees that servers are not 
overloaded, since we prove that it is possible to guarantee stability in this case.

%
\subsection{Related Work}
For many years, the common belief was that only overloaded queues\footnote{A \cyan{server queue is considered to be overloaded when the total arrival rate at the} server is greater than the service rate.} could generate instability, while underloaded ones could only induce delays that are longer than desired, but always remain stable. This general wisdom goes back to the models of  networks originally developed by Kleinrock~\cite{kleinrock-75-a}, and based on Jackson queuing networks~\cite{Jac:Job}. Stability results for more general classes of queuing networks~\cite{321887,kelly1979} also confirmed that only overload generates instability. This belief was shown to be wrong when it was observed that, in some networks, the backlogs in specific queues could grow indefinitely even when such queues were not overloaded~\cite{lukumar91,rybko92}. 

Motivated by this fact, there has been an effort to understand the factors that can affect the stability of a queueing network. In~\cite{Chang94stability}, the authors provide some results regarding some conditions that render bounded queue lengths both for a single queue and for feedforward networks. In~\cite{10.2307/2245163}, it is shown a class of networks which, although queues are served substantially more quickly that the rate at which tasks are injected, their mean service times are as small as desired. By using a simple queueing network, in~\cite{10.2307/4140498} it was shown that conditions on the mean interarrival and service times are not enough to determine its stability under a particular policy.

It was later shown that instability could also arise in some types of Kelly networks~\cite{AndrewsAFLLK01,BorodinKRSW01} (a network is said to be of the \emph{Kelly type}~\cite{kelly1979} when servers have the same service rates). These results in the \emph{Adversarial Queuing Theory} (AQT) model aroused an interest in understanding the stability properties of packet-switched  networks. This has attracted the attention of many researchers in recent years (see, for instance, the results in~\cite{blesa2006stability,1285116,caro2008performance,tsaparas1999stability}).

\subsection{Our Work}
In this paper, we introduce a model
to analyze queuing systems of computational jobs formed by dependent tasks.
We call this model \emph{Adversarial Job Queueing} (AJQ). The main 
contribution of the AJQ model is a novel approach for modeling 
\emph{tasks} and their dependencies in the computational \emph{job} 
system which is much richer than the modelling capabilities of AQT. 
As mentioned, a job is composed by
a set of tasks, each described by some parameters, like the server in 
which the task must be executed or the time that the task needs to 
be completed.
Additionally, each task depends on other tasks of the same job 
(i.~e., subsets of tasks that must be completed before the given task starts). 

\cyan{
The rich variety of task dependencies that we allow is, as far as we know, unique of our formalism, and makes AJQ very adaptive to model a variety of complex scenarios (including AQT as a special case).
For instance, our model allows imposing that a task $q$ cannot start until
a set $P$ of other tasks of the same job are completed. This expresses a
scenario in which task $q$ aggregates the results obtained by the tasks in
$P$. One example of this configuration is a MapReduce computation \cite{dean2008mapreduce}, in which the reduce task has to wait for all the map tasks to complete. This dependence in which the task $q$ needs \emph{all} the tasks in set $P$ to complete is called an \emph{AND dependence.} However, our formalism allows for more expressiveness by means of the \emph{OR dependence}, in which several AND dependencies are combined. In this case, a task $q$
has several sets $P_1, P_2, \ldots, P_l$, and it waits for \emph{any} set $P_i$ to be completed. This configuration appears, for instance, when several redundant tasks are used, so that the output of any of them is equally valid as input for $q$ \cite{gomes2001algorithm}.
}

As mentioned, tasks are processed in servers. \cyan{When tasks are 
active (ready to be executed) at a server 
but not being processed yet, they are maintained in a queue}
at the server. It is assumed that each server has an infinite buffer to 
store its own queue of active tasks. We use a bounded adversarial 
setting in our model. In this setting, 
we assume that an adversary injects jobs in the system, choosing the 
time and the characteristics of each injected job, with certain limits. 
This leads to worst-case system analyses. 
A desirable property under this model is that each server's service rate 
matches its injection rate for an arbitrarily long period of time, which 
implies the \emph{stability} of the system, 
since the number of jobs at any time is bounded. 

In the rest of this paper, we define the model more formally and provide 
some results regarding both the stability and instability under different
assumptions. From the point of view of the dependencies between tasks,
we show that if they are \emph{feed-forward} (see below) then the system 
is stable. From the point of view of the scheduling policies (i.~e., how 
a server decides which task to execute next), we observe that, since AJQ 
is more general than AQT (once we do the appropriate matching between jobs and packets, 
and between links and servers) unstable scheduling policies in AQT are 
easily translated into policies that are unstable in AJQ. 
On the other hand, we show that some stable 
scheduling policies in AQT remain stable in AJQ. For instance,
we prove that \textit{LIS}, which
gives priority to older \cyan{tasks/packets} (and is stable in AQT for any rate below 1), 
is stable in AJQ if the injection rate of jobs is below
a certain value that depends on the tasks processing time and activation delay.
Finally, we show that there are other policies
that are stable in AQT but unstable in AJQ.



\section{Model}\label{sec:model}
In this section, we define the \emph{Adversarial Job Queueing} (AJQ) 
model. The AJQ model is designed to analyze systems of queueing jobs. 
The three main components of an AJQ system $(S,P,\cA)$ are:
\begin{itemize} 
\item a set $S=\{s_1,s_2,\ldots,s_{n}\}$ of $n$ servers,
\item an \emph{adversary} $\cA$ who injects jobs in the system, and
\item a scheduling \emph{policy} $P$, which is the criteria used by 
servers to decide which task to serve next among the 
tasks waiting in their queues.
\end{itemize} 
The system evolves over time continuously (unlike AQT, which assumes discrete time). In each moment, 
the adversary may inject jobs to the system while the servers process
those jobs. In each moment as well, 
some tasks 
may be waiting to be executed, others may be in process, and others
may be completed. A job is considered \cyan{\emph{completed}} when all its tasks 
are completed. When a job is \cyan{completed,} all its tasks disappear from 
the system. 

Each job $\langle K, f^K \rangle$ consists of a finite set $K$ of tasks 
and a function $f^K$ that determines dependencies among the tasks.
(For simplicity we will denote the job $\langle K, f^K \rangle$ by its task set $K$.)
Let $K = \{k_1,k_2,k_3, \ldots,k_{l_K}\}$ be a job, 
where each 
$k_i$ is a task of $K$. The integer
$l_K$ denotes the number of tasks 
of $K$.  Each \cyan{task} $k_i$ is 
\cyan{defined} by 
three parameters $\langle s_i^K,d_i^K,t_i^K \rangle$. The parameter
$s_i^K \in S$ 
\cyan{is}
the server in which $k_i$ must be executed. 
The parameter \cyan{$d_i^K \geq 0$} is the \emph{activation delay} of $k_i$.  
The parameter 
\cyan{$t_i^K > 0$} is the \emph{processing time} of $k_i$, 
i.~e., the time server $s_i^K$
takes to execute task $k_i$. 

\red{Let $(S,P,\cA)$ be an AJQ system.
Let $T_{max}:=\max_{i,K} \{t_i^K\}$ and $T_{min}:= \min_{i,K} \{t_i^K\}$
be the maximum and minimum time, respectively, 
required to complete a task of any job $K$ injected in the system. 
We assume that these two quantities are \cyan{bounded} and do not depend on 
the time. 
Let $D_{min}:= \min_{i,K} \{d_i^K\}$ and 
$D_{max}:= \max_{i,K} \{d_i^K\}$ be the minimum and maximum 
activation 
delay, respectively, 
among all tasks of any job injected in the system. 
Since, $d_i^K \geq 0$ it follows that 
$D_{min}\geq 0$. On the other hand, we assume that $D_{max}$ is 
a constant that may depend on the parameters of the system, but it does 
not change over time. 
Finally, we use $L=\max_K \{l_K\}$ to denote the maximum number of tasks (length) of a 
job, which we assume is also a constant that does not depend 
on the time.} 

{\bf Feasibility.} Let $\mathcal{P}(K)$ be the \emph{power set} of $K$, i.~e., the 
set of all subsets of $K$. Furthermore, let  $\mathcal{P}^2(K)$ 
be the \emph{second power set} of $K$, i.~e., 
the set of all subsets of $\mathcal{P}(K)$. 
Given a job $K$, a \emph{feasibility function} 
$f^K:K\rightarrow \mathcal{P}^2(K)$  
determines which tasks of $K$ are 
\cyan{\emph{feasible}, which means that they are ready to be executed, once the activation delay has passed.}
Let $f^K(k_i)$ be equal to $\{A_1,A_2,\ldots, A_{\ell_i}\}$.
The sets $A_x$ for $1\leq x \leq \ell_i$ are called \emph{feasibility sets} for $k_i$.
Then, the 
task $k_i$ is \cyan{\emph{feasible} at a time $t$ if there exists 
a feasibility set $A_x$ for $k_i$
such that all tasks in $A_x$
have been completed by time $t$.}  
Otherwise, $k_i$ is \emph{blocked,} 
\cyan{and still has to wait for some other 
tasks of $K$ 
to complete before becoming feasible.}

The activation delay $d_i^k$ of a task $k_i$ represents a setup cost, 
expressed in time, that $k_i$ must incur once it becomes feasible
and before it \cyan{can start} to be processed. If $t$ is the \cyan{time
instant at 
which $k_i$ becomes} feasible, then $k_i$ will incur its activation delay 
during time 
interval $[t,t+d_i^k]$. Hence, it cannot be executed during such 
interval,
\cyan{in which} we say that task $k_i$ is a \emph{delayed}
feasible task (or only delayed task). When $k_i$ completes its 
activation delay \cyan{at time $t+d_i^k$, it can be served,
and since that moment} will be referred to as an 
\emph{active} feasible task, or simply active task. 
Equivalently, a feasible task is 
active if it has been feasible for at least $d_i^k$ time.
A job with at least one feasible \cyan{(resp., active)} task will be referred to as a 
\emph{feasible} \cyan{(resp., \emph{active})} job. 

\red{With the feasibility function, 
a task cannot start being served until some given
state of the tasks in the same job holds. Hence, the feasibility 
function can be used, for instance, to force the execution sequence of the 
tasks of a job. It enhances the modeling capabilities of the AJQ model by allowing the coexistence of AND
dependencies and OR dependencies,}
\cyan{as mentioned.} \\

{\bf Doability.} Let $K$ be a job and $k_i$ be a task of $K$. We say that $k_i$
is an \emph{initial} task of $K$ if $\emptyset \in f^K(k_i)$.
\cyan{Observe that all initial tasks $k_i$ are automatically feasible at the time the job $K$ is injected, and they become active $d^K_i$ time later.}

\cyan{We assign a \emph{layer} $\lambda(K,i)$ to the tasks $k_i$ of a job $K$ 
as follows.} 
All initial tasks have layer $\lambda(K,i)=1$. 
For any $j>1$, a task $k_i$ is assigned layer $\lambda(K,i)=j$ if
it is not feasible when all tasks of layers $1,...,j-2$ 
are completed, but it becomes feasible when additionally the tasks of
layer $j-1$ are completed. Let $\lambda_K\leq l_K$ \cyan{denote} the number of layers of job $K$.
If a task $k_i$ has layer \cyan{$\lambda(K,i)=\ell$, then there is a feasibility set 
$A_x \in f^K(k_i)$ for $k_i$ such that $A_x \subseteq \{k_j \in K: \lambda(K,j) < \ell \}$.}

\cyan{Observe that the above definition does not guarantee that all tasks of a job will be assigned a layer. In fact, it is not hard to create jobs that have tasks dependencies (e.g., cyclic dependencies) that prevent some tasks from being assigned a layer. Table~\ref{table:notdoable} shows an example of a job 
whose tasks get layer numbers and an example with tasks that cannot be assigned a layer number.
}

\begin{table}[t]
\begin{minipage}{0.5\linewidth}
     \centering
  \caption{This table shows the feasibility function of two 
   different jobs $J$ and $K$ defined over the same set of tasks 
   $\{1,2,3,4,5\}$. 
   Column $f^J(\cdot)$ shows the feasibility function of job $J$, while column $f^K(\cdot)$
   shows the feasibility function of job $K$.
   \blue{Job $J$ is not doable, since tasks $2, 3, 4$ and $5$
   \cyan{cannot be assigned a layer. On the other hand,} job $K$ is doable. Indeed, the number of each task in job $K$
   corresponds to its layer.}
   }
     \label{table:notdoable}
    \begin{tabular}{|c|c|c|}
\hline
Tasks & $f^J(\cdot)$    & $f^K(\cdot)$      \\ \hline
$1$   & $\{\emptyset\}$ & $\{\emptyset\}$   \\ \hline
$2$   & $\{\{1,5\}\}$   & $\{\{1\},\{5\}\}$ \\ \hline
$3$   & $\{\{2\}\}$     & $\{\{2\}\}$       \\ \hline
$4$   & $\{\{3\}\}$     & $\{\{3\}\}$       \\ \hline
$5$   & $\{\{4\}\}$     & $\{\{4\}\}$       \\ \hline
\end{tabular}
\end{minipage}\hfill
\begin{minipage}{0.5\linewidth}
\centering
    \includegraphics[width=0.6\textwidth]{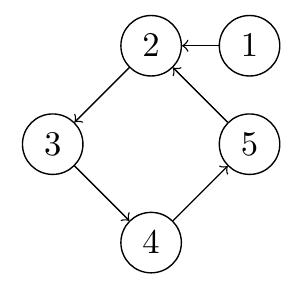}
    \caption{This figure shows the skeleton of jobs $J$ and $K$ presented in Table \ref{table:notdoable}, which are the same.}
    \label{fig:skeleton}
    \end{minipage}
\end{table}


We want every job to be potentially \cyan{completed}. Therefore, we impose some
restrictions over every feasibility function. 
\begin{definition}\label{def:doabletasks}
Let $K$ be a job and $f^K:K\rightarrow \mathcal{P}^2(K)$  be its
feasibility function. We say that $K$ is \emph{doable} if
every task $k_i$ of $K$ can be assigned a layer. 
\end{definition}
It is worth mentioning that, deciding whether a job is doable or not \cyan{as defined} can be computed in polynomial time
with respect to the size of the job (that takes into account the number of tasks and the size of the 
feasibility function).  \cyan{Indeed, layer 1 can be computed by checking which tasks have the empty set 
as a feasibility set. Then, a simple recursive algorithm computes all tasks in layer $i$
using all the tasks in layers $1,2,\ldots,i-1$.}

We show in the next proposition, that the \cyan{condition of doable job is necessary for a job to be completed, and that it is also sufficient if it is the only job injected in a system and the scheduling policy is work conserving}.
\begin{proposition}
Let $(S,P,\cA)$ be a system where the adversary $\cA$ injects only one job $K$
\cyan{and $P$ is work conserving. Then,}
$K$ can be completed if and only if $K$ is doable. 
\end{proposition}
\begin{proof}
On one hand, if $K$ is doable, $K$ \cyan{can be completed, since until that happens, there will always be at least one feasible 
task not completed.} To see this, assume by contradiction that there is a moment before $K$ is completed such that no task is feasible. Consider \cyan{any task $k_i$ among those that have not} been \cyan{completed} yet 
with the smallest layer (since $K$ is doable, all tasks have a layer). 
Therefore, $k_i$ has a feasibility set that is a subset of the \cyan{completed tasks}. Hence, $k_i$ is feasible, which is a contradiction. 
\cyan{Then, since there are always feasible tasks, their activation time is bounded, there are no other jobs in the system, and $P$ is work conserving, eventually all tasks of $K$ will become active, be scheduled and processed, and complete.}

On the other hand,
\cyan{assume that job $K$ is not doable but completes all its tasks in system $(S,P,\cA)$.
Then, all tasks in $K$ become feasible at some point in time, even those that are not assigned a layer. Consider the first task $k_i$ that becomes feasible among those that have no layer (break ties randomly). If this happens at time $t$, let $U$ be the set of tasks that completed by time $t$, and let $\ell=\max_{k_j \in U} \lambda(K,j)$. Then, from the procedure to assign layers to tasks, $k_i$ would have been assigned a layer $\lambda(K,i)\leq \ell+1$, which is a contradiction.
}
\end{proof} 

{\bf Topologies.} Let $K$ be a job, and $k_i$ and $k_j$ be two tasks of $K$. 
We say that $k_i$ \emph{depends} on $k_j$ if there exists 
a feasibility set $A_x\in f^K(k_i)$ for $k_i$
such that $k_j \in A_x$.

\begin{definition}
The \emph{skeleton} 
of a job~$K$ is
the directed graph $H_K=(V,E)$, where 
$V(H_K):=\{k_1,k_2,\ldots,k_{l_K}\}$ and $E(H_K) :=\{(k_j,k_i): k_i \mbox{ depends 
on } k_j\}$.
\end{definition}

It is worthwhile to mention that a skeleton does not define the feasibility function of a job. The 
two jobs presented in Table \ref{table:notdoable} are 
different jobs on the same set of tasks and with the same skeleton
(see Figure \ref{fig:skeleton}).
Nevertheless, one of the two jobs in Table \ref{table:notdoable} is doable and the other is not.  
Hence, the skeleton does not even
differentiate between doable and not doable jobs.

\red{The \emph{topology} of a job~$K$ is the directed graph 
obtained by mapping the skeleton of $K$ into the set of servers, where 
each task $k_i$ is mapped into its corresponding server $s_i^K$. 
\begin{definition}
Given a system $(S,P,\cA)$, the \emph{topology of the system} is the directed graph obtained by overlapping the topology of all jobs injected by $\cA$ in the system.
\end{definition}

Figures \ref{fig:skeletonjob1} and \ref{fig:skeletonjob2} show the skeleton of two jobs whose feasibility functions are described in tables \ref{table:job1} and \ref{table:job2}.
Figures \ref{fig:skeletonjob1} and \ref{fig:skeletonjob2} also show the layers of the jobs.
The topology of a system in 
which only those two jobs are injected is shown in Figure \ref{fig:topologysystem}.}

\begin{table}
\begin{minipage}{0.45\linewidth}
\centering
\includegraphics[width=0.9\linewidth]{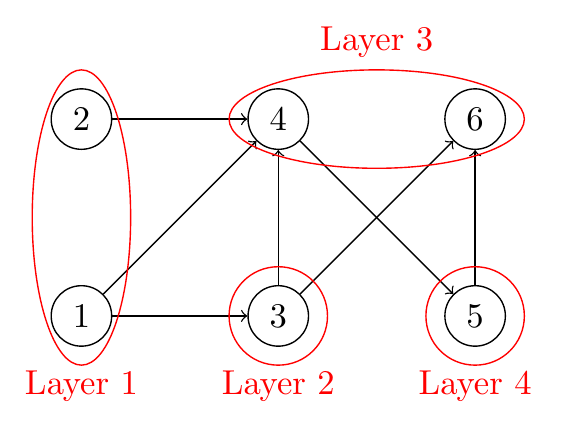}
\caption{Skeleton of job $R$ with set of tasks $\{1,2,3,4,5,6\}$. The layers of the job are circled in red. 
}
\label{fig:skeletonjob1}
\end{minipage}\hfill
\begin{minipage}{0.45\linewidth}
\centering
\caption{This table shows the feasibility function of job $R$ and the servers to which the tasks of job $R$ are assigned. }
\label{table:job1}
\begin{tabular}{|c|c|c|}
\hline
\begin{tabular}[c]{@{}c@{}}Task $i$ \\ of job $R$\end{tabular} & $f^R(i)$          & $s^R_i$ \\ \hline
$1$                                                            & $\{\emptyset\}$       & $s_1$   \\ \hline
$2$                                                            & $\{\emptyset\}$       & $s_1$   \\ \hline
$3$                                                            & $\{\{1\}\}$           & $s_1$   \\ \hline
$4$                                                            & $\{\{1,2,3\}\}$       & $s_2$   \\ \hline
$5$                                                            & $\{\{4\}\}$           & $s_3$   \\ \hline
$6$                                                            & $\{\{3\},\{5\}\}$ & $s_4$   \\ \hline
\end{tabular}
\end{minipage}
\vspace{2cm}

\begin{minipage}{0.45\linewidth}
\centering
\includegraphics[width=0.9\linewidth]{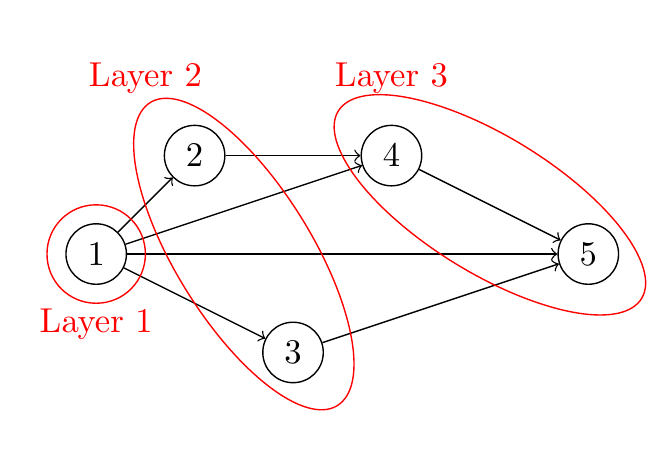}
\caption{Skeleton of job $M$ with set of tasks $\{1,2,3,4,5\}$. The layers of the job are circled in red. 
}
\label{fig:skeletonjob2}
\end{minipage}\hfill
\begin{minipage}{0.45\linewidth}
\centering
\caption{This table shows the feasibility function of job $M$ and the servers to which the tasks of job $M$ are assigned. }
\label{table:job2}
\begin{tabular}{|c|c|c|}
\hline
\begin{tabular}[c]{@{}c@{}}Task $i$ \\ of job $M$\end{tabular} & $f^M(i)$          & $s^M_i$ \\ \hline
$1$                                                            & $\{\emptyset\}$       & $s_4$   \\ \hline
$2$                                                            & $\{\{1\}\}$       & $s_3$   \\ \hline
$3$                                                            & $\{\{1\}\}$           & $s_3$   \\ \hline
$4$                                                            & $\{\{1,2\}\}$       & $s_2$   \\ \hline
$5$                                                            & $\{\{1,3\},\{4\}\}$           & $s_1$   \\ \hline
\end{tabular}
\end{minipage}

\vspace{2cm}

\centering
\begin{tikzpicture}[->,node distance=3cm,auto]

\node[state](A) {$s_1$};
\node[state] (B) [right of = A] {$s_2$};
\node[state] (C) [right of = B] {$s_3$};
\node[state] (D) [right of = C] {$s_4$};

\path (A) edge [loop above] (A)
	edge (B)
	edge [bend right] (D)
	(B) edge (C)
	(C) edge (D)
	(D) edge [bend right, dashed] (C)
	(D) edge [bend left, dashed ] (B)
	(C) edge [bend right, dashed] (B)
	(D) edge [bend right, dashed] (A)
	(C) edge [bend right, dashed] (A)
	(B) edge [bend right, dashed] (A)
	;
\end{tikzpicture}
\caption{This figure shows the topology of job $R$ (described in Figure \ref{fig:skeletonjob1} and Table \ref{table:job1}) in solid lines, the topology  of job $M$ (described in Figure \ref{fig:skeletonjob2} and Table \ref{table:job2}) in dashed lines, and, with all the lines, the topology of a system in which only these two jobs are injected by the adversary.}
\label{fig:topologysystem}
\end{table}



{\bf Scheduling policy.}
We assume that each server has an infinite buffer to 
store its own queue of tasks.
Every \cyan{active task waits}
in the queue of its corresponding server.
In each server, a scheduling policy $P$ specifies which task 
of all active tasks in its queue to serve next. We assume that
scheduling 
policies are greedy/work conserving (i.~e., a server 
always decides to serve if there is at least one active task 
in its queue).
Examples of policies are \emph{First-In-First-Out} (\textit{FIFO}) which gives 
priority to the task that first came in the queue, or 
\emph{Last-In-First-Out} (\textit{LIFO}) which gives priority to the task that 
came last in the queue. Other policies will be defined later in the 
document. 

{\bf Adversary.}
We assume that there is a malicious adversary $\cA$ who injects \emph{doable}
jobs into the system.
In order to avoid trivial overloads, the adversary is bounded in the 
following way. Let $N_s(I)$ be the total load injected by the adversary 
during  
time interval $I$ in server $s$ (i.~e., $N_s(I) = \sum t_i^K$ over 
all jobs $K$ injected during $I$ and tasks $k_i$ such that $s_i^K = 
s$). Then, for every server $s$ 
and interval $I$ 
the adversary is bounded by:
\begin{equation}\label{advbound}
N_s(I) \leq r|I| + b,
\end{equation}
where $0<r \leq 1$ is called the \emph{injection rate}, and \cyan{$b > 1$}
is called the \emph{burstiness} allowed to the adversary. Observe that 
(\ref{advbound}) implies $\max_{i,K} \{t_i^K\} \leq b$, since jobs are
injected 
instantaneously. 
An adversary that satisfies (\ref{advbound}) is called
a \emph{bounded} $(r,b)$-adversary, or simply an $(r,b)$-adversary. 

\cyan{As mentioned,} the system formed by an $(r,b)$-adversary $\cA$ injecting doable jobs in the set of 
servers $S$ using the scheduling policy $P$ is called an AJQ system 
$(S,P,\cA)$. 

The number of active tasks
in the queue of server $s$ at time $t$ is denoted $Q_s(t)$.
\begin{definition}
Let $(S,P,\cA)$ be an AJQ system. 
We say that the system $(S,P,\cA)$ is \emph{stable} if there exists a value
$M$ such 
that $Q_s(t) \leq M$ for all $t$ and for all $s \in S$, where $M$ may 
depend on the system parameters (adversary, servers, and jobs 
characteristics) but not on the time.
\end{definition}

\begin{definition}
Let $P$ be a policy. 
If a system $(S,P,\cA)$ is stable against any $(r,b)$-adversary $\cA$ with rate $r <
1$, then we say that the policy $P$ is \emph{universally stable}.
\end{definition}


In the next sections, we provide some results regarding both the 
stability and instability in the AJQ model.

\section{Stability and Instability of Scheduling Policies}

From the point of view of the scheduling policies (i.~e., how a server 
decides which task to \cyan{choose} from the set of active tasks pending to be 
executed), in this section we show stability of the policy that gives
priority to
the task (job) that has been for the longest period of time in the system.
On the other hand, we show that other well-known scheduling policies are not stable.


\subsection{Stability of \textit{LIS}}

The \textit{LIS} (Longest-In-System) scheduling policy gives priority to the
task (and hence the job) which has been in the system for the longest time. 
In this subsection, we show that any system $(S, \textit{LIS},\cA)$ is stable, 
for any $(r,b)$-adversary $\cA$ with $r< T_{min}/(T_{max}+D_{max})$.
We start by showing a bound on the time that a job spends in the system
until it is done.

Consider a job $K=\{k_1,k_2,\ldots,k_{l_K}\}$ injected at time $T_0$.
Let $T_i$ be the first time in which all tasks in the
$i$-\emph{th} layer of $K$ are completed. 
The time $T_{\lambda_K}$ is the time
when $K$ is done. 
Let $T$ be some time in the interval $[T_0,T_{\lambda_K}]$. 
We denote by $g_T$ the injection
time of the oldest job \cyan{that is still} in the system at time $T$. We define 
\[c := \max_{T\in [T_0,T_{\lambda_K}]} \{ T -g_T \}.\]

\begin{lemma}
Let $(S, \textit{LIS},\cA)$ be an AJQ system where $\cA$ is an
$(r,b)$-adversary with $r< T_{min}/(T_{max}+D_{max})$. Then,
\[
T_{\lambda_K}-T_0\leq \left(D_{max}+\frac{r(c+b)}{T_{min}}(T_{max}+D_{max})\right).
\]
\end{lemma}

\begin{proof}
Let $K$ be a job.
Let $k^*$ be the last task to be processed in the 
$i$-\emph{th} layer of $K$.
Hence, $k^*$ is complete at time $T_i$.
All tasks in the $i$-\emph{th} layer of $K$ become feasible by 
time $T_{i-1}$, including $k^*$.  From 
definition of $c$, only tasks injected in the interval
$[T_{i-1}-c,T_0]$ can block $k^*$ in its server. The tasks injected 
in this interval, including all the tasks in the $i$-\emph{th}
layer of $K$, are at most $r(T_0 - T_{i-1} + c + b)/T_{min}$. 
All these tasks are processed in at most $r(T_0 - T_{i-1} + c +
b)(T_{max}+D_{max})/T_{min}$ time. Hence:
\begin{eqnarray*}
T_i &\leq& T_{i-1} + D_{max} + \frac{r(T_0 - T_{i-1} + c + 
b)}{T_{min}}(T_{max}+D_{max})\\
&=&T_{i-1}\left(1 - \frac{r(T_{max}+D_{max})}{T_{min}}\right)
+ D_{max} + \frac{r(T_0+c+b)}{T_{min}}(T_{max}+D_{max})
\end{eqnarray*}

Let $\epsilon:=1 - r(T_{max}+D_{max})/T_{min}$. 
Solving the recurrence, we obtain:
\begin{eqnarray*}
T_{\lambda_K} &\leq& \epsilon^{\lambda_K}T_0 + \left(D_{max} + 
\frac{r(T_0+c+b)}{T_{min}}(T_{max}+D_{max})\right)\sum_{i=0}^{\lambda_K-1
}\epsilon^i\\
&=& \epsilon^{\lambda_K}t_0 + \left(D_{max} + 
\frac{r(T_0+c+b)}{T_{min}}(T_{max}+D_{max})\right)\left( 
\frac{1-\epsilon^{\lambda_K}}{1-\epsilon}\right)\\
&=& \left(D_{max}+\frac{r(c+b)}{T_{min}}(T_{max}+D_{max})\right) +T_0.
\end{eqnarray*} 
Which proves the lemma. 
\end{proof}

Since we are considering a case where $r< T_{min}/(T_{max}+D_{max})$, 
it holds that $r(T_{max}+D_{max})/T_{min} = 1-\epsilon < 1$. Hence, 
we rewrite the lemma as follows:
\[
T_{\lambda_K}-T_0\leq  (1-\epsilon)c +  \left(D_{max}+\frac{rb}{T_{min}}(T_{max}+D_{max})\right).
\]

\begin{theorem}
\label{thm:lis}
Let $(S, \textrm{LIS},\cA)$ be an AJQ system where $\cA$ is an
$(r,b)$-adversary with $r< T_{min}/(T_{max}+D_{max})$.
Then, all jobs spend less than
\[
\left(\frac{D_{max}T_{min}+rb(T_{max}+D_{max})}{T_{min} - 
r(T_{max}+D_{max})}\right)
\]
time in the system.
\end{theorem}
\begin{proof}
It is worth mentioning that $c$ is the only time-depending parameter 
in the bound given by the previous lemma. Hence, if we show that
$c$ actually does not depend on time, we will be showing the theorem. 
We prove it by contradiction. Assume that there is a moment in which $c$
is strictly larger than:
\[
\left(\frac{D_{max}T_{min}+rb(T_{max}+D_{max})}{T_{min} - 
r(T_{max}+D_{max})}\right).
\]
Hence, there has been a job in the system for a period of time strictly
longer than:
\[
\left(\frac{D_{max}T_{min}+rb(T_{max}+D_{max})}{T_{min} - 
r(T_{max}+D_{max})}\right).
\]
If we apply the previous lemma to this job, it should have been absorbed
in at most:
\begin{eqnarray*}
&&(1-\epsilon)c +  \left(D_{max}+\frac{rb}{T_{min}}(T_{max}+D_{max})\right)
\\
&=& c - \epsilon \left(\frac{D_{max}T_{min}+rb(T_{max}+D_{max})}{T_{min} - 
r(T_{max}+D_{max})}\right) + \left(D_{max}+\frac{rb}{T_{min}}(T_{max}+D_{max})\right)\\
&<& c
\end{eqnarray*}
time, which is a contradiction. 
\end{proof}

%

\subsection{Scheduling Policies that are Unstable}

Here, we show that a number of well-known policies such as 
First-In-First-Out (\textit{FIFO}), Nearest-To-Go (\textit{NTG}), Furthest-From-Source 
(\textit{FFS}), and Last-In-First-Out (\textit{LIFO}), are unstable, even for arbitrarily small injection rates.
While the meaning of \textit{FIFO} and \textit{LIFO} in the context of AJQ is clear (and similar as in AQT), we need to define \textit{NTG} and \textit{FFS}.

For a task $k_i$ of job $K$ the \emph{distance from source} is 
the distance between the layer of $k_i$ and layer one
($\lambda(K,i)-1$), and the distance to go is the distance between the number of layers of $K$ and $k_i$'s layer ($\lambda_K - \lambda(K,i)$).
Hence, \textit{FFS} gives priority to the task with largest distance from source and \textit{NTG} gives priority to the task with smallest distance to go.

\begin{theorem}
\label{the:unstable}
\textit{FIFO}, \textit{NTG}, \textit{FFS}, and \textit{LIFO} are unstable for every $r>0$.
\end{theorem}
\begin{proof}
First, we highlight that, given a system $(G,P,\cA)$ in AQT, it can be 
modeled as a system $(S,P,\cA')$ in AJQ as follows:

\begin{itemize}
\item
For each link $l$ in $G$, there is a unique server $s_l$ in $S$, which we call its \emph{equivalent} server.
\item
The scheduling policy $P$ is the same  both in AQT and in AJQ.
\item
For each packet $p$ injected by $\cA$, the adversary $\cA'$ injects a 
job $K$ such that:
\begin{itemize}
    \item
    For each link $l$ in the path of packet $p$, there is a task $k_l$ in $K$ to be executed in server $s_l$.
    \item
    If $l$ is the first link in the path of $p$, then $k_l$ is the initial task of job $K$.
    \item
    If the path of packet $p$ traverses link $l$ \cyan{immediately} before it traverses link $l'$ then task $k_{l'}$ \cyan{only} depends on task $k_{l}$.
    \item 
    The processing time of each task is $1$ and its activation delay
    is $0$.
\end{itemize}

\end{itemize}
Clearly, if $(G,P,\cA)$ is unstable for a given injection rate then 
$(S,P,\cA')$ will be also unstable for the same injection rate (i.~e., 
all the unstable scheduling policies in AQT are also unstable in AJQ).

By using the results in~\cite{BorodinKRSW01}, we have that \textit{NTG}, \textit{FFS},  and \textit{LIFO} are unstable (in AQT) for every $r>0$, and by using the result in~\cite{1039916} (in AQT) we have that \textit{FIFO} is also unstable for every $r>0$. Therefore, the theorem directly follows. 
\end{proof}

\section{Topological Stability}
\label{l-topology}
In this section we show stability for systems with feed-forward topology.
We say that a system has \emph{feed-forward} topology
if it is possible to enumerate the servers from $1$ to $n$, so that
every directed arc in the topology of the system goes from a server 
with a smaller label to a server with a larger label.

\begin{theorem}
Let $(S,P,\cA)$ be an AJQ system with feed-forward topology. 
Then, for any policy $P$ and any 
$(r,b)$-adversary $\cA$ with injection rate 
$r \leq 1$, the system $(S,P,\cA)$ is stable.
\end{theorem}

\begin{proof}
Let $(S,P,\cA)$ be an AJQ system with feed-forward topology. 
Without loss of generality, assume that the ordering of the set of 
servers that makes the system feed-forward is $s_1,s_2,\ldots,s_n$.
For simplicity, we only use the position $j$ to denote server $s_j$.
Let $\tau_j(t)$ be the time that server $j$ would
need to completely serve (drain) all its pending tasks present
in the system at time $t$ if they were all active and no new 
task were injected:
\[
\tau_j(t) := \sum_{K(j,t)} t_i^K,
\]
where $K(j,t)$ is the set of pairs $(K,i)$ such that $K$ was 
injected by time $t$, $s_i^K=j$, and task $k_i$ has not been completed in server $j$. 
%
We define a potential function $\Phi(\cdot)$ as follows:
\begin{eqnarray*}
\Phi(0):= D_{max} + b; \ \ \ 
\Phi(1):= \tau_1(0) + \Phi(0) + b,
\end{eqnarray*}
where $\tau_j(0)$ denotes the time server $j$ requires to process all
its tasks present 
in the system before the adversary starts injecting jobs in
the system. For $2\leq j\leq n$, $\Phi(j)$ is defined as:
\[
\Phi(j) = \tau_j(0) + \Phi(0)+\frac{\sum_{i=1}^{j-1} \Phi(i)}{T_{min}}\cdot 
L\cdot (T_{max}+D_{max})+b,
\]  


To prove this theorem, we show that for all $1\leq j \leq n$
and for all $t$,
$\tau_j(t)\leq \Phi(j).$
We use induction on $j$,
the position of the servers in the ordering of $S$.

{\bf Case $j=1$:}
Consider some time $T$. We prove that 
$\tau_1(T)\leq \Phi(1)$.
First, assume that for all time $t \in [0,T]$ there is at least 
one active task in the queue of server $1$. 
Then, server $1$ has been continuously working during the interval 
$[0,T]$. 

The time required to process all its queue at time $T$ 
is the time it would need to process the tasks present at time $0$ in 
its queue, plus the time required to process the load injected by the 
adversary during that interval, minus the load processed during that 
interval. In the form of an equation, the previous amount of time is: 
\begin{eqnarray*}
\tau_1(T) \leq \tau_1(0) + T + b - T = \tau_1(0) + b \leq \Phi(1).
\end{eqnarray*}

Otherwise, there exists some time $t \in [0,T]$ such that
there is no active task in the queue of server $1$ at time $t$. 
Let $t^*$ be the largest of such times. Note that, in that case, 
all tasks injected in server $1$ before time $t^*$, and present at 
time $t^*$, were injected after time $t^*-D_{max}$, since every task
injected before that time is active at time $t^*$. Therefore, 
by restriction (\ref{advbound}), it holds:
$\tau_1(t^*)\leq D_{max} + b$. 

Then, the time required to process all its queue at time $T$ 
is the time it would need to process the tasks present at time 
$t^*$, plus the time required to process the load injected by the 
adversary during the interval $[t^*,T]$, 
minus the load processed during the same  
interval of time. 
In the form of an equation, the previous amount of time is:
\begin{eqnarray*}
\tau_1(T) \leq \tau_1(t^*) + (T-t^*) + b - (T-t^*)
\leq D_{max} + b + b
= \Phi(0) + b
\leq \Phi(1).
\end{eqnarray*}

{\bf Case $j>1$:}
The inductive hypothesis is $\tau_{i}(t) \leq \Phi(i)$ for all 
for all $t$ and for all $1\leq i < j$.
By inductive hypothesis then, 
the amount of tasks in servers $i<j$ is at most
$
\frac{\sum_{i=1}^{j-1}\Phi(i)}{T_{min}},
$
the number of tasks they can trigger in server $j$ is at most
$
\frac{\sum_{i=1}^{j-1}\Phi(i)}{T_{min}} \cdot L,
$
and the processing time for all those tasks is at most
$
\frac{\sum_{i=1}^{j-1}\Phi(i)}{T_{min}} \cdot L\cdot (T_{max}+D_{max}).
$

Consider again some time $T$. 
We consider two cases equivalent to those considered in the case $j=1$.
First, server $j$ has at least one active task in its queue during all
the interval $[0,T]$. Therefore, server $j$ has processed tasks 
during all 
that time. In that case, server $j$ would need all the time required to 
process the tasks present at time $0$ in its queue, plus all the time 
required to process the tasks triggered by tasks in previous 
servers, plus all the time required to process the load injected by the 
adversary during the interval $[0,T]$, minus the load processed during 
that interval. Which, in the form of an equation is: 
\begin{eqnarray*}
\tau_j(T) &\leq& \tau_j(0) + \frac{\sum_{i=1}^{j-1}\Phi(i)}{T_{min}} \cdot 
L\cdot (T_{max}+D_{max})+ T + b - T \\
 &=& \tau_j(0) + \frac{\sum_{i=1}^{j-1}\Phi(i)}{T_{min}} \cdot L\cdot 
 (T_{max}+D_{max})+  b \leq \Phi(k).
\end{eqnarray*}

Assume now that there is some time $t \in [0,T]$ such that 
there is no active task in the queue of server $j$ at time $t$. 
Let $t^*$ be the largest such time. An analysis equivalent to 
the one presented in the case $j=1$ shows that 
$\tau_j(t^*)\leq D_{max} + b$.   
Therefore, if we compute $\tau_j(T)$ equivalently to the previous cases,
we obtain: 
\begin{eqnarray*}
\tau_j(T) &\leq& \tau_j(t^*) + \frac{\sum_{i=1}^{j-1}\Phi(i)}{T_{min}} \cdot 
L\cdot (T_{max}+D_{max})+ (T-t^*) + b - (T-t^*)\\ 
&\leq& D_{max} + b + \frac{\sum_{i=1}^{j-1}\Phi(i)}{T_{min}} \cdot 
L\cdot (T_{max}+D_{max}) + b \leq \Phi(k).
\end{eqnarray*}

Hence, 
the time required by server $j$ to drain its queue 
is bounded by $\Phi(j)$, a function that does not depend on time. 
In conclusion, at any time, there are at most $\Phi(j)/T_{min}$ tasks 
in the queue of server $j$, and the system is stable. 
\end{proof}

We showed that a 
feed-forward topology is a sufficient condition for stability in a 
system. Nevertheless, this condition is not necessary. Indeed,
as we have shown before, for any system $(G,P,\cA)$ in the AQT model, 
there is an equivalent system $(S,P,\cA')$ in the AJQ model. 
We know that, in the AQT model, any system with a ring network 
(i.~e., a directed cycle) is stable with any scheduling policy and 
against any adversary. Therefore, the equivalent system in the AJQ model 
will also be stable. Nevertheless, such AJQ system has a topology that 
is not feed-forward.

\section{Job Properties that can Affect the Stability of the System}

In this section, we show that some of the features of the injected jobs 
can play a key role regarding the stability of the system. Namely, we 
show that both the tasks' processing time and activation delays are 
factors that, individually, can cause instability. We also show that 
the feasibility function can lead, by itself, to instability.

\subsection{Tasks' Processing Times}
We show that the processing times of the tasks can affect the stability of
the system. Namely, a stable system can be transformed into unstable by 
varying the processing time of some of their tasks, even if the adversary has the same rate $r$ in both systems.
Let \textit{LCT-LIS} be the scheduling policy that gives priority to the task 
with longest processing time at the current server, breaking ties 
according to the longest-in-system policy.

\begin{proposition}
\label{propos:LET-LIS}
There exists a server set $S$ and an adversary $\cA$ with injection rate 
$r>1/\sqrt{2}$ such that the system $(S,\textit{LCT-LIS},\cA)$ is 
unstable.
\end{proposition}
\begin{proof}
The proof is inspired by the instability by difference in packet length 
proof in the \emph{continuous} AQT~\cite{BlesaCFLMSST09} (CAQT) model. 
Let $(G,\textit{LPL-LIS},\cA')$ be the system used in Theorem~26 
in~\cite{BlesaCFLMSST09} (\textit{LPL-LIS}  denotes the scheduling policy that  
gives priority to the packets with longest length, breaking ties 
according to the longest in system policy). Note that 
$(G,\textit{LPL-LIS},\cA')$ can be seen as an AQT system, except that two
different packet lengths ($1$ and $2$) are taken into account. 
Let us now consider a system $(S,\textit{LCT-LIS},\cA)$ in AJQ, such 
that:
\begin{itemize}
    \item 
    The scheduling policy is \textit{LCT-LIS}.
    \item
    For each packet $p$ injected by $\cA'$, the adversary $\cA$ injects a
    job such that all its tasks have a processing time ($1$ and $2$)
    equal to the length of the injected packet. 
    \item
    The rest of the system is modeled in the same fashion as in the proof of 
    Theorem~\ref{the:unstable}.
\end{itemize}
%
Theorem 26 in~\cite{BlesaCFLMSST09} shows that $(G,\textit{LPL-LIS},\cA')$ is unstable for an injection rate 
 $r>1/\sqrt{2}$. Therefore, it is not hard to derive that $(S,\textit{LCT-LIS},\cA)$ is 
 also unstable for the same rate.  
\end{proof}

Figure~\ref{fig:LET-LIS} illustrates the system $S$ used in the proof of the previous proposition and provides some details about its unstable behavior.


\begin{figure}[t]
  \centering
    \includegraphics[width=1\textwidth]{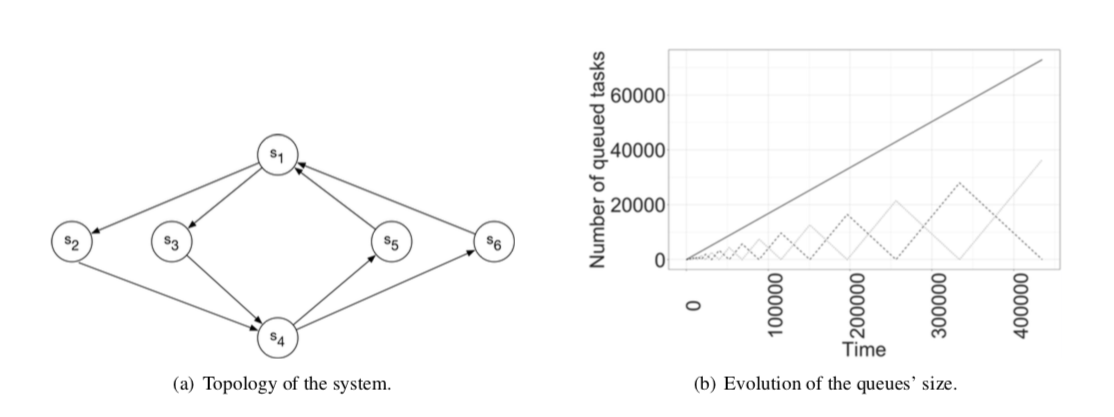}
   \caption{System used in the proof of Proposition~\ref{propos:LET-LIS}. Subfigure~(a) shows the system's topology and subfigure~(b) shows the number of queued packets at each time instant for an injection rate of $0.8$ (which is slightly higher than $1/\sqrt{2}$): dashed lines correspond to servers $s_1$ and $s_4$, and the solid line corresponds to the overall system. As it can be seen, the number of queued \cyan{tasks} at $s_1$ and $s_4$ \emph{oscillate} in an alternating and increasing fashion, which provokes a continuous increase in the system's number of queued \cyan{tasks}.}
  \label{fig:LET-LIS}
\end{figure}


Observe that if all tasks have the same processing time $T$, then the \textit{LCT-LIS} scheduling policy becomes \textit{LIS}.
As shown in Theorem~\ref{thm:lis}, \textit{LIS} is stable for any $r< T_{min}/(T_{max}+D_{max})=T/(T+D_{max})$. Hence, for small $D_{max}$ (e.g., $D_{max}=0$), we have a rate $r>1/\sqrt{2}$ for which \textit{LCT-LIS} is stable if all tasks have the same processing time. Therefore, we have shown that an unstable system can be transformed into stable by only varying the processing times of some of their tasks.

\subsection{Tasks' Activation Delays}

As it has been done in the previous subsection, here we show that the activation delays of the tasks can affect the stability of the system.
Let \textit{SAD-NFS} be the scheduling policy that gives priority to the task with smallest activation delay at the queue of the current server, breaking ties according to the nearest from source policy regarding to an initial task in the job's skeleton.

\begin{proposition}
There exists a server set $S$ and an adversary $\cA$ with injection rate  $r>1/\sqrt{2}$ such that the system $(S,\textit{SAD-NFS},\cA)$ is unstable.
\end{proposition}
\begin{proof}
The proof follows the lines of the one in Proposition~\ref{propos:LET-LIS}.
Let $(G,\textit{SPP-NFS},\cA')$ be the system used in Theorem~28 in~\cite{BlesaCFLMSST09} (\textit{SPP-NFS}  denotes the scheduling policy that  gives priority to the packets whose previously traversed link had smallest propagation delay, breaking ties according to the nearest-from-source policy). Note that in this system the transmission time of every packet is the same in every link. Hence, $(G,\textit{SPP-NFS},\cA')$ can be seen as an AQT system, except that some links have a positive fixed propagation delay. 

Let us now consider a system $(S,\textit{SAD-NFS},\cA)$ in AJQ, such that:
\begin{itemize}
    \item 
    The scheduling policy is \textit{SAD-NFS}.
    \item
    For each link $l$ in $G$ with a propagation delay $d_l$, all tasks executed in its equivalent server will have an activation delay equal to $d_l$. That is, the activation delays are seen as the delays taken by packets to traverse the links (besides the times spend at the queues).
    \item
    The rest of the system is modeled in the same fashion as in Theorem~\ref{the:unstable}.
\end{itemize}

Clearly, if the system $(G,\textit{SPP-NFS},\cA')$ is unstable for a given injection rate then $(S,\textit{SAD-NFS},\cA)$ will be also unstable for the same injection rate. However, by using the result in~\cite{BlesaCFLMSST09} (Theorem 28), we have that $(G,\textit{SPP-NFS},\cA')$ is unstable for an injection rate $r>1/\sqrt{2}$. Therefore, we have that $(S,\textit{SAD-NFS},\cA)$ is also unstable for the same rate. 
\end{proof}

Note that if all links in the system $(G,\textit{SPP-NFS},\cA')$ of the previous proof have zero delay it becomes an AQT system, 
and the $(S,\textit{SAD-NFS},\cA)$ system obtained has only tasks with activation delay of $0$. 
In that case, both \textit{SPP-NFS} and \textit{SAD-NFS} behave as \textit{NFS} in their respective systems.
Moreover, since \textit{NFS} is universally stable in AQT as shown in~\cite{AndrewsAFLLK01}, both systems $(G,\textit{SPP-NFS},\cA')$ and $(S,\textit{SAD-NFS},\cA)$ are stable.
Hence, we have shown that an unstable system can be transformed into stable by only varying the activation delays of some of their tasks.

\subsection{Feasibility Function Among Tasks}

Now, we show that the feasibility function is a factor that, by itself, can also induce instability.
We say that a feasibility function is \emph{fully independent} if no task in any job depends on any other task (i.~e.,
all tasks are initial). In this case, we also say that the tasks are fully independent.

\begin{proposition}
Let $(S,P,\cA)$ be an AJQ system such that all the tasks are fully independent. 
Then, for any set of servers $S$, any policy $P$ and any 
$(r,b)$-adversary $\cA$ with injection rate 
$r \leq 1$, the system $(S,P,\cA)$ is stable.
\end{proposition}
\begin{proof}
Direct, from the injection bound of Equation~(\ref{advbound}) and the fact that $P$ is work conserving. 
\end{proof}

Then, it is clear that if we take an unstable system and make all 
tasks fully independent, it will become stable.

\section{Future Work}
The AJQ model opens interesting research questions. Regarding scheduling
policies, it is still unknown whether there exists a universally stable 
policy \cyan{(i.~e., a policy stable under any adversary with $r<1$)}. Indeed, all the parameters of the model make difficult to 
see the existence of a universally stable policy. Regarding systems' 
topology, a full characterization of the topologies that produce a stable 
system against any bounded adversary is still open. \cyan{For instance, while
we argue in Section~\ref{l-topology} that the universal stability of the ring
in AQT can be propagated to AJQ, it is only for jobs that mimic the dependencies and topology of AQT.
It would be interesting to know whether all AJQ systems with a ring topology are
stable under bounded adversaries.}

On another hand, the AJQ model can be extended transferring the resource
allocation decision from the adversary to the scheduling policy. 
In that case, the adversary could provide, for each task, a 
set of servers in which it can be processed
(instead of a single server, as it is done
in our model). In that extended model, we would be able to study the
impact of resource allocation into the stability of a system.

\bibliographystyle{plain}
\bibliography{radio-route,tesis}

\end{document}